\documentclass{article}
\usepackage[utf8]{inputenc}

\usepackage{amsmath,amssymb,bbm,amsthm}
\usepackage{fullpage}
\usepackage{thm-restate,color,xcolor,xspace}
\usepackage{hyperref,cleveref}
\usepackage{algorithm,algorithmic}
\usepackage{thm-restate}
\usepackage{graphicx}
\usepackage[center]{caption}

\newcommand{\mycomment}[1]{}

\newtheorem{theorem}{Theorem}[section]
\newtheorem{lemma}[theorem]{Lemma}

\newtheorem{observation}[theorem]{Observation}

\title{Listing 4-Cycles}
\author{Amir Abboud\footnote{Weizmann Institute of Science. Supported by an Alon scholarship and a research grant from the Center for New Scientists at the Weizmann Institute of Science. Email: \text{amir.abboud@weizmann.ac.il}} \and Seri Khoury\footnote{UC Berkeley. Email: \text{seri\text{\_}khoury@berkeley.edu}} \and Oree Leibowitz\footnote{Weizmann Institute of Science. Email: \text{oree.leibowitz@weizmann.ac.il}} \and Ron Safier\footnote{Weizmann Institute of Science. Email: \text{ron.safier@weizmann.ac.il}}}
\date{}

\begin{document}

\maketitle

\begin{abstract}
   In this note we present an algorithm that lists all $4$-cycles in a graph in time $\tilde{O}(\min(n^2,m^{4/3})+t)$ where $t$ is their number. Notably, this separates $4$-cycle listing from triangle-listing, since the latter has a $(\min(n^3,m^{3/2})+t)^{1-o(1)}$ lower bound under the $3$-SUM Conjecture.
   Our upper bound is conditionally tight because (1) $O(n^2,m^{4/3})$ is the best known bound for detecting if the graph has any $4$-cycle, and (2) it matches a recent $(\min(n^3,m^{3/2})+t)^{1-o(1)}$ $3$-SUM lower bound for enumeration algorithms.
   The latter lower bound was proved very recently by Abboud, Bringmann, and Fischer [arXiv, 2022]  and  independently by Jin and Xu [arXiv, 2022].
   
   \paragraph{Independent work:} Jin and Xu [arXiv, 2022] also present an algorithm with the same time bound.
\end{abstract}

\section{Introduction}

Finding small patterns in large graphs is a classical task.
Perhaps the two smallest patterns that make such problems non-trivial are the triangle and the $4$-cycle.
The best known algorithms for \emph{detecting} if a graph has at least one pattern take $O(\min(n^{\omega},m^{2\omega/(\omega+1)}))$ for triangle \cite{alon1997finding}, where $\omega<2.37188$ is the fast matrix multiplication exponent \cite{AlmanW21,https://doi.org/10.48550/arxiv.2210.10173}, and $O(\min(n^2,m^{4/3}))$ for $4$-cycle \cite{alon1997finding}.\footnote{Throughout, we assume that graphs are unweighted, undirected, and have $n$ nodes and $m$ edges.} Note that if $\omega=2$ then the two bounds are the same.

In the \emph{listing} formulation we are asked to return all occurrences of the pattern in the graph. 
Simple exhaustive search lets us list all triangles in $O(\min(n^3,m^{1.5})+t)$ time, where $t$ is their number.
A more clever algorithm by Björklund, Pagh, Vassilevska William, and Zwick \cite{bjorklund2014listing} has an upper bound of 
$\Tilde{O}(n^\omega+n^\frac{3(\omega-1)}{5-\omega} t^\frac{2(3-\omega)}{5-\omega})$ and $\Tilde{O}(m^\frac{2\omega}{\omega+1}+m^\frac{3(\omega-1)}{\omega+1} t^\frac{3-\omega}{\omega+1})$. Assuming $\omega=2$ the running times simplify to $\Tilde{O}(n^2+nt^{2/3})$ and $\Tilde{O}(m^{4/3}+mt^{1/3})$ which is essentially tight under  the $3$-SUM Conjecture due to a reduction of Kopelowitz, Pettie, and Porat \cite{Kopelowitz2016HigherLB} that optimizes a construction by Pătraşcu \cite{patrascu2010towards}, and by the APSP Conjecture by a reduction of Vassilevska Williams and Xu \cite{Williams2020MonochromaticTT}.

We present the first algorithm improving over exhaustive search for $4$-cycle listing. Notably, it is faster than the lower bounds for triangle listing and therefore separates the two problems (under the $3$-SUM Conjecture).

\begin{theorem}
$4$-cycle listing can be solved in $O(\min(n^2+t,(m^{4/3}+t)\cdot \log^2 n))$ time.
\end{theorem}

Any improvement on our upper bound would break the longstanding upper bound for $4$-cycle \emph{detection}.
The first super-linear lower bound under popular conjectures for $4$-cycle listing was proved recently by Abboud, Bringmann, Khoury, and Zamir \cite{ABKZ22} using the ``short cycle removal'' technique.
Recent work by Abboud, Bringmann, and Fischer \cite{ABFarxiv} and independently by Jin and Xu \cite{JXarxiv} optimized this technique and proved a $(\min(n^2,m^{4/3})+t)^{1-o(1)}$ lower bound under the $3$-SUM Conjecture. Our upper bound complements these lower bounds.\footnote{These lower bounds only hold for the closely related \emph{enumeration} problem; our upper bounds also apply in that setting by standard techniques.}

\paragraph{Independent work:} Jin and Xu \cite{JXarxiv} independently obtained the same result by a similar technique.


\section{Upper bounds for 4-Cycle Listing}
In Section~\ref{sec:warmup}, we start with a simple extension of the folklore $O(n^2)$-time algorithm for $4$-cycle detection~\cite{RICHARDS1985249} to an $O(n^2+t)$-time algorithm for $4$-cycles listing, where $t$ is their number.

Then, in Section~\ref{sec:4/3Alg}, we present our main result in this note, which is an $\tilde{O}(m^{4/3}+t)$-time algorithm for $4$-cycles listing.

\subsection{Warm-up: An $O(n^2+t)$ algorithm}\label{sec:warmup}

\begin{observation}\label{thm:n-listing}
Given a graph $G$, there is an $O(n^2+t)$-time algorithm that lists all the $4$-cycles, where $t$ is their number.
\end{observation}
\begin{proof}[Proof Sketch]
We can list all $2$-paths, by going over all pairs of nodes $u,v$ and list the $2$-paths between $u$ and $v$. After listing all $2$-paths, we go over all pairs of nodes $u,v$ again and list the $4$-cycles that $u$ and $v$ participate in, by going over all pairs of $2$-paths between $u$ and $v$. 

The time complexity is linear in the number of $2$-paths and the number of $4$-cycles $t$. By observing that the number of $2$-paths is $O(n^2+t)$, the claim follows. 
\end{proof}

\subsection{An $\tilde{O}(m^{4/3}+t)$-time algorithm}\label{sec:4/3Alg}

In order to improve the $O(n^2+t)$-time algorithm, we can't afford listing all $2$-paths. For instance, in a star graph, there are $O(n^2)$ $2$-paths, but no $4$-cycles. Hence, intuitively speaking, we need to narrow our attention to a certain type of $2$-paths that are useful for listing $4$-cycles efficiently. Indeed, to overcome the star example, it suffices to note that there is no point in listing $2$-paths with endpoints of degree one (leaves), as these $2$-paths can't be extended to $4$-cycles. 

To extend this intuition, perhaps one could try to split the nodes into low-degree and high-degree groups, denoted by $L$ and $H$, respectively, and consider different types of $2$-paths with respect to this partitioning. The advantage of such partitioning is that we can narrow our attention to specific types of $2$-paths that are more challenging for listing. For instance, one can immediately spot two types of $2$-paths that are less challenging for listing. The first is the type of $2$-paths with a low-degree node at the center, and the second is the type of $2$-paths that use only high-degree nodes. For listing the first type, we just need to go over all the low-degree nodes, and list their pairs of neighbors, and there are only $O(\Delta^2)$ such pairs per low-degree node, where $\Delta$ is the degree threshold. For the second type, we can bound the number of high-degree nodes by $2m/\Delta$, which helps in bounding the number of $2$-paths that use only high-degree nodes. Indeed, by picking $\Delta=m^{1/3}$, listing these two types of $2$-paths takes $O(m^{4/3}+t)$ time (as shown in lemmas \ref{lem:all-h-2-paths} and \ref{lem:l-in-the-middle-2-paths}). Furthermore, listing these two types of $2$-paths suffices for listing all types of $4$-cycles, except the $4$-cycles that use two overlapping $2$-paths of the form $LHH$ ($2$-paths with a high degree node at the center, one low degree endpoint, and one high degree endpoint). That is, these $4$-cycles are of the form $LHHL$. To list these $4$-cycles, we need to find a way to list $2$-paths of the form $LHH$.

Unfortunately, we can't afford listing all $2$-paths of the form $LHH$. For instance, take a graph where there is a node $u$ that is connected to $n$ leaves (low degree nodes), and to $n^{2/3-\epsilon}$ high-degree nodes, where each of these high-degree nodes is connected to $n^{1/3+\epsilon}$ leaves. In this example, we have $m=O(n)$ edges, $m^{5/3-\epsilon}\gg m^{4/3}$ $2$-paths of the form $LHH$ (the ones that go through $u$), but no $4$-cycles. 

To overcome such examples, recall that the only remaining type of $4$-cycles that we need to list are the ones of the form $LHHL$. Since we know how to list $2$-paths of the form $HLL$ ($2$-paths with a low-degree node at the center) efficiently, it suffices to list \emph{only one of the two} $LHH$ $2$-paths that such a $4$-cycle consists of. Hence, for the $4$-cycles of the form $LHHL$, one could wonder: is there a property that one of the two overlapping $LHH$ paths (that the $4$-cycle consists of) must have, that would make it easier to list such $4$-cycles? 

Indeed, given a $4$-cycle of the form $LHHL$, for the two middle high-degree nodes, one of them must have degree greater or equal to the other. Therefore, it suffices to list $2$-paths of the form $LHH$, where the degree of the middle node is at most the degree of the third node (the high-degree endpoint). We refer to such $2$-paths as $L\rightarrow H\rightarrow H$ (the orientation from a node $u$ to a node $v$ here means that $u$'s degree is at most $v$'s). The question that remains is: can we afford listing all $2$-paths of the form $L\rightarrow H\rightarrow H$? Interestingly, in this note we answer this question affirmatively. Roughly speaking, we show that there can't be that many $L\rightarrow H\rightarrow H$ paths compared to $4$-cycles. Hence, we use a charging argument that allows us to list all such $2$-paths.

\paragraph{A road-map for the technical parts.} First, in Section~\ref{sec:LHH}, we prove a helpful theorem that shows that there can't be too many $2$-paths of the form $L\rightarrow H\rightarrow H$ compared to $4$-cycles. We refer to this theorem as the $L\rightarrow H\rightarrow H$ theorem. Then, in Section~\ref{sec:alg}, we put everything together and prove our main result - an $\tilde{O}(m^{4/3} + t)$-time algorithm for $4$-cycle listing.

\subsection{The $L\rightarrow H\rightarrow H$ Theorem}\label{sec:LHH}
In this section we prove the following theorem that connects the number of $4$-cycles in a graph to the number of $2$-paths of a certain type. The degree of a node $v$ is denoted by $\deg(v)$.
\begin{theorem}\label{thm:conj}
Given an undirected graph $G=(V,E)$ with $m$ edges, let $H$ be the set of nodes with degree larger than $m^{1/3}$, and $L=V\setminus H$. Orient all the edges $\{u,v\}$ from $u$ to $v$ if $\deg(u)\leq \deg(v)$ (break ties arbitrarily). Let $P$ be the number of directed $2$-paths of the form $L\rightarrow H\rightarrow H$. It holds that if $P>100m^{4/3}\log^2n$, then there are at least $P/(100\log^2n)$ $4$-cycles.
\end{theorem}
In order to proof theorem \ref{thm:conj}, we use two helper lemmas. In Lemma~\ref{lem:AvgDegreeCycles}, we show that the number of $4$-cycles is $\Omega(d^4-n^2)$, where $d$ is the average degree.\footnote{This statement follows from known techniques; Jin and Xu use \cite{BGSW19} who proved it using a convexity argument. We give a self-contained short proof with a spectral argument that is morally similar and might be of independent interest.} In Lemma~\ref{lem:regularView}, we provide a view of the graph that has some nice properties. In particular, this view is a partitioning of the graph that has a useful regularity property, while the number of $L\rightarrow H\rightarrow H$ $2$-paths is preserved. The proof of Theorem~\ref{thm:conj} in provided after the proof of Lemma~\ref{lem:regularView}.
\begin{lemma}\label{lem:AvgDegreeCycles}
Any graph with $n$ nodes and average degree $d$ has $\Omega(d^4-n^2)$ $4$-cycles.
\begin{proof}
Let $G=(V,E)$ be a graph with $n$ nodes and average degree $d$. Let $A$ be the adjacency matrix of $G$. Denote by $\lambda_1\ge \lambda_2 \ge \dots \ge \lambda_n$ the $n$ eigenvalues of $A$.
The top eigenvalue $\lambda_1$ of $A$ is at least $d$. This is because: 

$$\lambda_1 = \max_{v^T v = 1} v^T A v$$

Now, consider $u = (1/\sqrt{n},\dots,1/\sqrt{n})$. Clearly, it holds that $u^T u = 1$. On the other hand we have that:

$$u^T A u = \sum_{w\in V}\frac{1}{n} \deg(w) = 2|E(G)|/n = d$$

The number of closed 4-walks in $G$ is at least $d^4$. This is because this number is exactly the trace of $A^4$ and
$$\text{tr}(A^4)=\sum_{i=1}^n \lambda_i^4 \ge \lambda_1^4 = d^4$$

Let $T$ be the number of $4$-cycles. To finish the proof, we show that the number of closed walks is at most

\begin{align}\label{eq:2-walks}
    10n^2 + \frac{10 T}{{10\choose 2}}
\end{align}

This would imply that $T \geq d^4 - 10n^2 - \frac{10 T}{{10\choose 2}}$ which would imply that $T=\Omega(d^4 - n^2)$. To prove Equation~\ref{eq:2-walks}, it suffices to bound the number of $2$-paths. For this, let $S$ be the set of pairs of nodes that intersect at most $10$ times (that is, for each ${u,v}\in S$, we have $N(u)\cap N(v)\leq 10$). The number of $2$-paths can be bounded by 

\begin{align*}
    &\sum_{\{u,v\}\in S} N(u)\cap N(v) + \sum_{\{u,v\}\notin S} N(u)\cap N(v)\\
    \leq \hspace{0.1cm} & 10n^2 + \sum_{\{u,v\}\notin S} N(u)\cap N(v)\\
    \leq \hspace{0.1cm} & 10n^2 + \frac{10 T}{{10\choose 2}}
\end{align*}

where the last inequality holds because for any $\{u,v\}\notin S$, we can charge ${x\choose 2}$ $4$-cycles on only $x$ $2$-paths, where $x>10$. 

\end{proof}
\end{lemma}


\begin{lemma}\label{lem:regularView}
Given a graph $G=(V,E)$ with $m$ edges, let $H$ be the set of nodes with degree larger than $m^{1/3}$, and $L=V\setminus H$. Orient all the edges $\{u,v\}$ from $u$ to $v$ if $\deg(u)\leq \deg(v)$ (break ties arbitrarily). Let $P$ be the number of directed $2$-paths of the form $L\rightarrow H\rightarrow H$.
There is a partition of the set of nodes $H$ into two parts $A$ and $B$, such that the number of directed $2$-paths of the form $L\rightarrow A\rightarrow B$ is at least $P/(4\log^2n)$, and each node in $A$ has the same number of incoming edges from $L$ (up to a multiplicative $2$-factor), and the same number of outgoing edges to $B$ (up to a multiplicative $2$-factor). Furthermore, each node in $B$ has at least $1$ incoming edge from $A$.

\begin{proof}
We start by describing a simpler partitioning. This simpler partitioning splits the set of high-degree nodes $H$ into two sets $A$ and $B$ such that the number of $2$-paths of the form $L\rightarrow A\rightarrow B$ is at least $P/4$. Such a partitioning exists by the probabilistic method: Each node in $H$ joins $A$ with probability $1/2$ independently. Thus, the probability that a $2$-path $u\rightarrow v\rightarrow w$ that is of the form $L\rightarrow H\rightarrow H$ survives in $L\rightarrow A\rightarrow B$ is $\Pr(u\in A \wedge v\in B) = 1/4$. Hence, in expectation, we get $P/4$ such paths in $L\rightarrow A\rightarrow B$. \\

Next, let $\deg_{L}(u)$ be the number of incoming edges from $L$ to a node $u$. Similarly, $\deg_{ B}(u)$ is the number of outgoing edges to $B$. It holds that the number of $2$-paths of the form $L\rightarrow A\rightarrow B$ is 

$$P/4=\sum_{a\in A}\deg_{L}(a)\cdot \deg_{B}(a) = \sum_{(i,j)\in [\log n]} \sum_{\substack{a\in A \text{ s.t. }\\ \deg_L(a)\approx 2^i \\ \deg_B(a)\approx 2^j}}\deg_{L}(a)\cdot \deg_{B}(a)$$

Hence, one of the $(i,j)$ buckets is contributing at least $P/(4\log^2n)$ $2$-paths, as desired. Furthermore, we can trivially delete nodes in $B$ that have no incoming edges from $A$ (we mean the updated version of $A$ which is the set of nodes that fall into our large $(i,j)$ bucket).

\end{proof}
\end{lemma}

\begin{proof}[Proof of Theorem~\ref{thm:conj}]
First, take the partitioning from Lemma~\ref{lem:regularView}. We know that the number of $2$-paths of the form $L\rightarrow A\rightarrow B$ is $P'\geq P/(4\log^2n) \geq 25m^{4/3}$. Recall that each node in $A$ has the same in-degree from $L$ up to a multiplicative 2-factor. Denote by $d_L$ the minimum over these degrees. Similarly, each node in $A$ has the same out-degree to $B$ (up to a multiplicative 2-factor). Denote the minimum of these degrees by $d_B$. Furthermore, each node in $B$ has at least one incoming neighbor from $A$. 

We show that the number of $4$-cycles is $\Omega(P')$. For this, let $d_0$ be the average degree in the graph induced by the nodes in $A\cup B$. Since we have at most $m^{2/3}$ nodes in $A\cup B$, by Lemma~\ref{lem:AvgDegreeCycles}, it suffices to show that $d_0^4 = \Omega(P')$.

For this, we split the proof into two cases:

\begin{enumerate}
    \item $|A|>|B|$: Observe that in this case, $d_0 > d_B/2$. Hence, it sufficed to show that $d_B^4 = \Omega(P')$. For this, recall that $P' \leq 4 |A|\cdot d_L\cdot d_B$ (where the $4$ factor is coming from the two $2$ factors for $d_L$ and $d_B$), and assume towards a contradiction that $(d_B)^4<P'$ which implies (by substituting $d_B$ with $P'/ (4|A|\cdot d_L)$) that $(P')^3<(4|A|\cdot d_L)^4$. But this is impossible  because it would imply that $P'<16m^{4/3}$ (as $|A|\cdot d_L\leq m$).
    
    \item $|A|\leq |B|$: In this case, we have that $d_0 > d_B\cdot|A|/(2|B|)$. Hence, it suffices to show that $(d_B|A|/2|B|)^4 > P'$. Assume towards a contradiction that $(d_B|A|/2|B|)^4 \le P'$. By substituting $d_B$ with $P'/(4d_L\cdot |A|)$, this implies that $(P')^3<(8d_L\cdot |B|)^4$. Now we want to argue that $d_L\cdot |B|$ is at most $m$ to get a contradiction to $P' \ge 25m^{4/3}$. For this, recall that each node in $B$ has at least one incoming edge from $A$, which implies that the degree of each node in $B$ is at least $d_L$. Hence, we have that $m\geq \sum_{u\in B} \deg(u)\geq \sum_{u\in B} d_L=|B|\cdot d_L$, as desired. 
\end{enumerate}
\end{proof}

\subsection{Listing $4$-cycles}\label{sec:alg}
In this section we prove the following theorem.

\begin{theorem}\label{thm:m-listing}
Listing all the $4$-cycles in an undirected graph $G=(V,E)$ can be done in ${O}(m^{4/3}\log^{2}n + t \log^{2}n)$ time, where $m$ is the number of edges and $t$ is the number of $4$-cycles.
\end{theorem}
The proof of Theorem~\ref{thm:m-listing} is based on listing several types of $2$-paths efficiently. Each of the lemmas \ref{lem:all-h-2-paths}, \ref{lem:l-in-the-middle-2-paths} and \ref{lem:l-h-h-oriented-2-paths} shows that we can list a certain type of $2$-paths efficiently.
\begin{lemma}\label{lem:all-h-2-paths}
Given a graph $G=(V,E)$ with $m$ edges. Let $H$ be the set of nodes with degree larger than $m^{1/3}$ and $L=V\setminus H$. Listing all the $2$-paths with only nodes from $H$ can be done in $O(m^{4/3}+t)$ time.
\end{lemma}
\begin{proof}
Let $G'$ be the subgraph of $G$ induced by $H$. Denote by $n'$ the number of nodes in $G'$ and by $t'$ the number of $4$-cycles in $G'$. Observe that $n'\le 2m^{2/3}$ and $t'\le t$. By using an argument similar to the one used in Observation~\ref{thm:n-listing}, we can list all the $2$-paths in $G'$ in time $O(n'^2 + t') = O(m^\frac{4}{3} + t)$.
\end{proof}
\begin{lemma}\label{lem:l-in-the-middle-2-paths}
Given a graph $G=(V,E)$ with $m$ edges. Let $H$ be the set of nodes with degree larger than $m^{1/3}$ and $L=V\setminus H$. Listing all the $2$-paths with a node from $L$ at the center can be done in $O(m^{4/3})$ time.
\end{lemma}
\begin{proof}
Scan all the edges in $G$ and for those with at least one endpoint in $L$ scan all the neighbors of the endpoints in $L$. This procedure finds all the $2$-paths with a node from $L$ in the middle.
Each node in $L$ has at most $m^{1/3}$ neighbors. Therefore, the running time of this procedure is $O(m^\frac{4}{3})$.
\end{proof}
\begin{lemma}\label{lem:l-h-h-oriented-2-paths}
Given a graph $G=(V,E)$ with $m$ edges. Let $H$ be the set of nodes with degree larger than $m^{1/3}$ and $L=V\setminus H$. Orient all the edges $\{u,v\}$ from $u$ to $v$ if $\deg(u)\leq \deg(v)$ (break ties arbitrarily). Listing all the directed $2$-paths of the form $L\to H\to H$ can be done in ${O}(m^{4/3} \log^{2}n + t \log^{2}n)$ time.
\end{lemma}
\begin{proof}

We can list all the $L\rightarrow H\rightarrow H$ $2$-paths in time that is linear in their number and the number of edges. This can be done by going over all the nodes $u\in L$, and then going over the neighbors $v\in H$ of $u$, and then going over all neighbors of $v$ with higher degree than $v$. This can be done in time that is linear in the number of edges and the number of $L\rightarrow H\rightarrow H$ $2$-paths. This is because we can prepare a set of higher degree nodes in $H$ for all the nodes $u\in V$, via a simple $O(m)$-time preprocessing step, where we go over all the edges (with at least one endpoint in $H$), detect for each edge the higher degree endpoint, and store it.
Therefore, the running time of the algorithm is $O(m + P)$ where $P$ is the number of directed $L\rightarrow H\rightarrow H$ $2$-paths. Since by Theorem~\ref{thm:conj} we have that $P=O(m^{4/3}\log^2n + t\log^2n)$, the claim follows.
\end{proof}


\begin{proof}[Proof of Theorem~\ref{thm:m-listing}]

We consider all the different types of $4$-cycles (in terms of low-high degree nodes) and show that we can least all of them in the desired running time.
\begin{description}
  \item[Type 1: $4$-cycles that use only high-degree nodes.] This class of $4$-cycles can be decomposed into two $2$-paths of all high-degree nodes. These $2$-paths can be listed in $O(m^{4/3}+t)$ time by Lemma~\ref{lem:all-h-2-paths}.
  \item[Type 2: $4$-cycles with two non-adjacent low-degree nodes.] This class of $4$-cycles can be decomposed into two $2$-paths with a low-degree node at the center. Theses $2$-paths can be listed in $O(m^{4/3})$ time by Lemma~\ref{lem:l-in-the-middle-2-paths}.
  \item[Type 3: $4$-cycles with three high-degree nodes and one low-degree node.] This class can be decomposed into a $2$-path of all high-degree nodes and a $2$-path with a low-degree node at the center. Using Lemma~\ref{lem:all-h-2-paths} and Lemma~\ref{lem:l-in-the-middle-2-paths} these $2$-paths can be listed in $O(m^{4/3}+t)$ time.
  \item[Type 4: $4$-cycles with two adjacent low-degree nodes and two adjacent high-degree nodes.] These $4$-cycles of the can be decomposed into a directed $L\rightarrow H\rightarrow H$ $2$-path and an $LLH$ $2$-path. Using Lemma~\ref{lem:l-h-h-oriented-2-paths} and Lemma~\ref{lem:l-in-the-middle-2-paths} we can list all these $2$-paths in ${O}(m^{4/3} \log^{2}n + t \cdot \log^{2}n)$ time.
\end{description}
To sum up, we showed how to list all the different types of $4$-cycles in ${O}(m^{4/3} \log^{2}n + t \cdot \log^{2}n)$ time, as desired.

\end{proof}
\begin{figure}[H]
\centering
\includegraphics[scale=0.2]{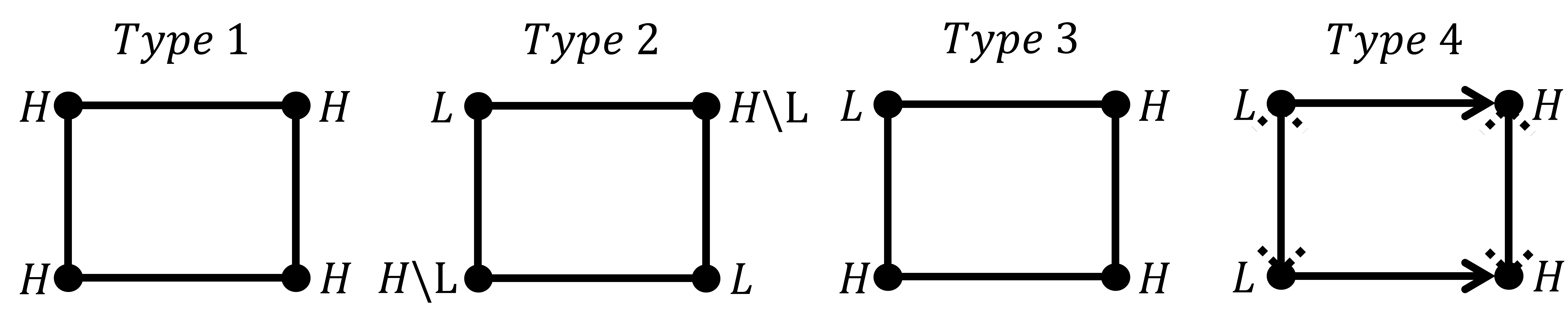}
\caption{The types of $4$-cycles. A $4$-cycle of the fourth type always consists an $L\rightarrow H\rightarrow H$ directed $2$-path and a $2$-path with a low-degree node at the center.}
\centering
\label{fig:4_cycles_types}
\end{figure}
\bibliographystyle{plain}
\bibliography{refs}
\end{document}